\documentclass[12pt]{iopart}

\pdfoutput=1

\usepackage{amsthm, amstext, amssymb, amsfonts, graphicx, bm}

\theoremstyle{definition}

\newtheorem{theorem}{Theorem}[section]
\newtheorem{lemma}[theorem]{Lemma}
\newtheorem{proposition}[theorem]{Proposition}
\newtheorem{definition}[theorem]{Definition}
\newtheorem*{remark}{Remark}

\renewcommand{\PR}{\mathbb P}

\begin{document}

\title{A cluster expansion approach to exponential random
graph models}

\author{Mei Yin}
\address{Department of Mathematics, University of Texas, Austin, TX
78712, USA} \ead{myin@math.utexas.edu}

\begin{abstract}
The exponential family of random graphs is among the most
widely-studied network models. We show that any exponential random
graph model may alternatively be viewed as a lattice gas model
with a finite Banach space norm. The system may then be treated by
cluster expansion methods from statistical mechanics. In
particular, we derive a convergent power series expansion for the
limiting free energy in the case of small parameters. Since the
free energy is the generating function for the expectations of
other random variables, this characterizes the structure and
behavior of the limiting network in this parameter region.
\end{abstract}

\maketitle

\section{Introduction}
\label{intro} Random graphs have been widely studied (see [1, 2]
for surveys of recent work) since the pioneering work on the
independent case. The first serious attempt was made by Solomonoff
and Rapoport \cite{SR} in the early 1950s, who proposed the
``random net'' model in their investigation into mathematical
biology. A decade later, Erd\H{o}s and R\'{e}nyi \cite{ER}
independently rediscovered this model and studied it exhaustively,
hence the namesake ``Erd\H{o}s-R\'{e}nyi random graph''. Their
construction was straightforward: Take $n$ identical vertices, and
connect each pair by undirected edges independently with
probability $p$. Many properties of this simple random graph are
exactly solvable in the large $n$ limit. Perhaps the most
important feature is that it possesses a phase transition: From a
low-density, low-$p$ state in which there are few edges to a
high-density, high-$p$ state in which an extensive fraction of all
vertices are joined together in a single giant component.

The Erd\H{o}s-R\'{e}nyi random graph, while illuminating, is a
poor model for most real-world networks, as has been argued by
many authors \cite{WS, SS, Do}, and so it has been extended in a
variety of ways. To address its unrealistic degree distribution,
generalized random graph models such as the configuration model
\cite{M} and the multipartite graph model \cite{NWS} have been
developed. However, these models have a serious shortcoming, in
that they fail to capture the common phenomenon of transitivity
exhibited in social and biological networks of various kinds.

The main hope for progress in this direction seems to lie in
formulating a model that incorporates graph structure in more
detail. A top candidate is the exponential random graph model, in
which dependence between the random edges is defined through
certain finite subgraphs, in imitation of the use of potential
energy to provide dependence between particle states in a grand
canonical ensemble of statistical physics. These exponential
models were first studied by Holland and Leinhardt \cite{HL} in
the directed case, and later developed extensively by Frank and
Strauss \cite{FS}, who related the random graph edges to a Markov
random field. More developments are summarized in \cite{WF, S, R}.

The past few years have witnessed a surge of interest in the study
of the limiting behavior of exponential random graphs. A major
problem in this field is the evaluation of the free energy, a
quantity that is crucial for carrying out maximum likelihood and
Bayesian inference. A particular motivation for people in the
statistical mechanics community to study the free energy is that
it provides information on phase transition in these sophisticated
models. Many people have made substantial contributions in this
area: H\"{a}ggstr\"{o}m and Jonasson \cite{HJ} examined the phase
structure and percolation phenomenon of the random triangle model.
Park and Newman \cite{PN1, PN2} constructed mean-field
approximations and analyzed the phase diagram for the
edge-two-star and edge-triangle models. Borgs et al. \cite{Bg}
established a lower bound on the largest component above the
critical threshold for random subgraphs of the $n$-cube.
Bollob\'{a}s et al. \cite{Bo} showed that for inhomogeneous random
graphs with (conditional) independence between the edges, the
critical point of the phase transition and the size of the giant
component above the transition could be determined under one very
weak assumption. Bhamidi et al. \cite{B} focused on the mixing
time of the Glauber dynamics and proposed that in the high
temperature regime the exponential random graph is not appreciably
different from the Erd\H{o}s-R\'{e}nyi random graph. Dembo and
Montanari \cite{DM} discovered that for the Ising model on a
sparse graph phase transitions and coexistence phenomena are
related to Gibbs measures on infinite trees. Using the emerging
tools of graph limits as developed by Lov\'{a}sz and coworkers
\cite{LS}, Chatterjee and Diaconis \cite{CD} gave the first
rigorous proof of singular behavior in the edge-triangle model.
They also suggested that, quite generally, models with repulsion
exhibit a transition qualitatively like the solid/fluid
transition, in which one phase has nontrivial structure, as
distinguished from the ``disordered'' Erd\H{o}s-R\'{e}nyi graphs.
Radin and Yin \cite{RY} derived the full phase diagram for a large
family of $2$-parameter exponential random graph models with
attraction and showed that they all contain a first order
transition curve ending in a second order critical point
(qualitatively similar to the gas/liquid transition in equilibrium
materials). Aristoff and Radin \cite{AR} considered random graph
models with repulsion and proved that the region of parameter
space corresponding to multipartite structure is separated by a
phase transition from the region of disordered graphs (proof
recently improved by Yin \cite{Yin}).

As is usual in statistical mechanics, we work with a finite
probability space, and interpret our results in some more
sophisticated limiting sense. We consider general $k$-parameter
families of exponential random graphs.
\begin{itemize}
\item A complete graph $K_n$ on $n$ vertices consists of a vertex
set $V_n$ ($|V_n|=n$) and an edge set $E_n$ ($|E_n|={n \choose
2}$). A vertex pair $e$ is a two-element subset of $V_n$, and the
set of all vertex pairs constitute $E_n$.

\item $\mathcal{G}_n$ is the set of simple graphs $G$ on $n$
vertices, where a graph $G$ with vertex set $V(G)=V_n$ is simple
if its edge set $E(G)$ is a subset of $E_n$.

\item $H_1,...,H_k$ are pre-chosen finite simple graphs. Each
$H_i$ has $m_i$ vertices ($2\leq m_i\leq m$) and $p_i$ edges
($1\leq p_i\leq p$). In particular, $H_1$ is $K_2$ (i.e. a single
edge).

\item A vertex map $f: V(H_i) \to V(G)$ is a homomorphism if the
induced edge map $f_*: E(H_i) \to E_n$ sends $E(H_i)$ into $E(G)$.
$t(H_i, G)$ is the density of graph homomorphisms $H_i \to G$:
\begin{equation}
\label{t} t(H_i, G)=\frac{|\text{hom}(H_i, G)|}{|\text{hom}(H_i,
K_n)|},
\end{equation}
where the denominator $|\text{hom}(H_i, K_n)|=n^{m_i}$ counts the
total number of mappings from $V(H_i)$ to $V(G)$.

\item $\beta=(\beta_1,...,\beta_k)$ are $k$ real parameters. They
tune the influence of the pre-chosen graphs $H_1,..., H_k$.
\end{itemize}

Let $T^\beta(G)$ be the weighted sum of graph homomorphism
densities $t(H_i, G)$:
\begin{equation}
T^\beta(G)=\sum_{i=1}^k \beta_i t(H_i, G).
\end{equation}
There is a probability mass function that assigns to every $G$
\begin{equation}
\label{pmf}
\PR_n^{\beta}(G)=\exp\left(n^2(T^{\beta}(G)-\psi_n^{\beta})\right),
\end{equation}
where $\psi^\beta_n$ is the normalization constant, i.e., it
satisfies
\begin{equation}
\label{psi} \sum_{G \in \mathcal{G}_n} \exp\left(n^2
T^\beta(G)\right)=\exp\left(n^2 \psi^\beta_n\right).
\end{equation}

The rest of this paper is organized as follows. In Section
\ref{model} we show that the general exponential random graph
model may alternatively be viewed as a lattice gas model with a
finite Banach space norm (Propositions \ref{stat} and
\ref{count}). This transforms the probability model into a
statistical mechanics model (Theorem \ref{con}). In Section
\ref{expansion} we apply cluster expansion techniques \cite{Y} and
derive a convergent power series expansion (high-temperature
expansion) for the limiting free energy in the case of small
parameters (Theorem \ref{main}). Finally, Section \ref{conclusion}
is devoted to concluding remarks.

\section{Alternative view}
\label{model} In this section we will transform the exponential
random graph model into a lattice gas model with a finite Banach
space norm. We begin by presenting an alternative view of the
homomorphism density $t(H_i, G)$ (\ref{t}).

\begin{definition}
\label{adj} Let $G \in \mathcal{G}_n$. Let $\sigma$ be the
indicator function of $E(G)$. For every vertex pair $e=\{i, j\}$,
\begin{equation*}
\sigma_{e}=\left\{%
\begin{array}{ll}
    1, & \hbox{an edge exists between vertices $i$ and $j$;} \\
    0, & \hbox{otherwise.} \\
\end{array}%
\right.
\end{equation*}
\end{definition}

\begin{definition}
\label{exact} Let $X \subseteq E_n$. Fix a finite simple graph $H$
with $p$ edges. Define the exact graph homomorphism density $d(H,
X)$ by
\begin{equation}
d(H, X)=\frac{|\text{ehom}(H, X)|}{|\text{hom}(H, K_n)|},
\end{equation}
where the numerator $|\text{ehom}(H, X)|$ counts the number of
homomorphisms $f: V(H) \to V_n$ whose induced map $f_*: E(H) \to
E_n$ sends $E(H)$ onto $X$. It is clear that $d(H, X)$ is
finite-body: $d(H, X)=0$ for $|X|>p$.
\end{definition}

\begin{proposition}
\label{stat} Let $G \in \mathcal{G}_n$. Let $\sigma$ be the
indicator function of $E(G)$. Fix a finite simple graph $H$. The
graph homomorphism density $t(H, G)$ has a lattice gas
representation
\begin{equation}
t(H, G)=\sum_{X \subseteq E_n}d(H, X)\sigma_X,
\end{equation}
where $\sigma_X=\prod_{e\in X}\sigma_{e}$.
\end{proposition}

\begin{proof}
By Definition \ref{exact}, the graph homomorphism density (\ref{t}) is given by
\begin{equation}
\label{alt} t(H, G)=\sum_{X \subseteq E(G)} d(H, X).
\end{equation}
Our claim easily follows once we realize that $\sigma_X=1$ if and
only if $X \subseteq E(G)$.
\end{proof}

\begin{figure}
\centering
\includegraphics[width=5in]{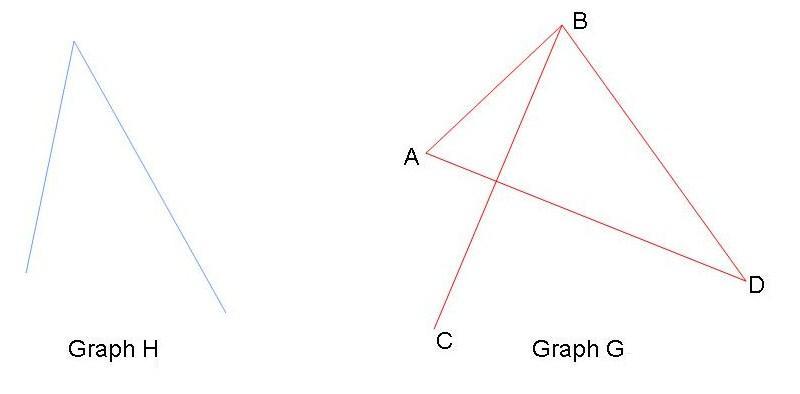}
\caption{Constructing a lattice gas representation for the graph
homomorphism density $t(H, G)$.} \label{weight}
\end{figure}

The above construction might be better explained with a concrete
example. Take $H$ a two-star (with $m=3$ vertices and $p=2$
edges). Let $G$ (see Figure \ref{weight}) be a finite simple graph
on $4$ vertices ($n=4$): $A$, $B$, $C$, and $D$. There are $18$
homomorphisms and $4^3=64$ total mappings from $V(H)$ to $V(G)$.
By (\ref{t}), the homomorphism density of $H$ in $G$ is $t(H,
G)=18/64$.

The density $t(H, G)$ may be derived through a lattice gas
representation as well. For notational convenience, we denote
$d(H, X)$ by $J(X)$. The image of $H$ in $G$ under a homomorphic
mapping is either an edge or a two-star. The exact homomorphism
density for an edge is $2/64$, and we have $J(\{A, B\})=J(\{A,
C\})=J(\{A, D\})=J(\{B, C\})=J(\{B, D\})=J(\{C, D\})=2/64$. The
exact homomorphism density for a two-star is also $2/64$, and we
have $J(\{A, B\}, \{B, C\})=J(\{B, A\}, \{A, C\})=J(\{A, C\}, \{C,
B\})=J(\{A, B\}, \{B, D\})=J(\{B, A\}, \{A, D\})=J(\{A, D\}, \{D,
B\})= J(\{A, C\}, \{C, D\})=J(\{C, A\}, \{A, D\})=J(\{A, D\}, \{D,
C\})= J(\{B, C\}, \{C, D\})=J(\{C, B\}, \{B, D\})=J(\{B, D\}, \{D,
C\})=2/64$. For all other $X \subseteq E_4$, we have $J(X)=0$. The
indicator function $\sigma$ of this particular $E(G)$ is given by:
$\sigma_{\{A, B\}}=\sigma_{\{A, D\}}=\sigma_{\{B,
C\}}=\sigma_{\{B, D\}}=1$ and $\sigma_{\{A, C\}}=\sigma_{\{C,
D\}}=0$. Therefore the valid images of $H$ in $G$ are: edges $AB$,
$AD$, $BC$, $BD$, each carrying density $2/64$; two-stars $(AB,
BC)$, $(AB, BD)$, $(CB, BD)$, $(BA, AD)$, $(AD, DB)$, each
carrying density $2/64$, making the combined density $\sum_{X
\subseteq E_4}J(X)\sigma_X=18/64$.

As $n$ gets large, it would seem hard to keep track of all the
densities $d(H, X)$ in the lattice gas representation, nevertheless, we will show that the
pinned densities have a universal upper bound.

\begin{proposition}
\label{count} Let $G \in \mathcal{G}_n$. Fix a finite simple graph
$H$ with $m$ vertices. Fix a vertex pair $e=\{i,j\}$ of $G$.
Denote by $t_{e}(H, G)$ the sum of the exact homomorphism
densities $d(H, X)$ with $e \in X \subseteq E(G)$. Then in the
lattice gas representation (cf. Proposition \ref{stat}), we have
\begin{equation}
t_{e}(H, G)=\sum_{X: e\in X \subseteq E(G)}d(H, X)\leq
\frac{m(m-1)}{n^2}.
\end{equation}
\end{proposition}

\begin{proof}
The homomorphisms under consideration must satisfy that the image
of $V(H)$ in $V(G)$ contain vertices $i$ and $j$ of $G$. To count
these homomorphisms, we regard such a mapping as consisting of two
steps. Step $1$: We construct vertex maps from $V(H)$ to $V(G)$:
First select two vertices of $H$ to map onto $i$ and $j$, of which
there are $m(m-1)$ choices. Then map the remaining vertices of $H$
onto $G$, of which there are $n^{m-2}$ ways. Step $2$: We check
whether these vertex maps are valid homomorphisms (i.e.
edge-preserving). The number of homomorphisms is thus bounded by
$m(m-1)n^{m-2}$. Our claim easily follows once we recall that the
total number of mappings from $V(H)$ to $V(G)$ is $n^m$.
\end{proof}

Our next theorem formulates the exponential random graph model as
a lattice gas model. The lattice is $E_n$, the set of all vertex
pairs $e$ of $V_n$. For each site $e=\{i, j\}$, we attach a
lattice gas variable $\sigma_{e}$ which takes on the value $1$ or
$0$. This specifies a simple graph $G$ on $V_n$ by Definition
\ref{adj}. The Hamiltonian $H(\sigma)$ is a weighted sum of graph
homomorphism densities $t(H_i, G)$ and varies according to the
structure of the graph $G$ (cf. Proposition \ref{stat}). It has a
finite Banach space norm which depends on the universal upper
bound for the pinned densities $d(H_i, X)$ (cf. Proposition
\ref{count}). This model is somewhat unconventional in the sense
that the underlying lattice is ``infinite-dimensional'' (a vertex pair
$e_1$ is a nearest neighbor of another vertex pair $e_2$ as long
as they have a vertex in common, so the number of
nearest neighbors of a given site grows with $n$, rather than being fixed at $2^d$ in a $d$-dimensional lattice), yet the associated interaction
is finite-body. A key interest in this model is its behavior in
the large $n$ limit. For that purpose, we will pay special
attention to the limiting free energy as it is the generating
function for the expectations of other random variables.

\begin{theorem}
\label{con} The general exponential random graph model (\ref{pmf})
is equivalent to a lattice gas model with a finite Banach space
norm.
\end{theorem}

\begin{proof}
The Hamiltonian is the negative of the exponent of the probability mass function
(\ref{pmf}) (without normalization):
\begin{equation}
H(\sigma)=-n^2 \sum_{i=1}^k \beta_i t(H_i, G),
\end{equation}
where $\sigma$ is the indicator function of $E(G)$. By Proposition
\ref{stat}, it has a lattice gas representation
\begin{equation}
H(\sigma)=-n^2\sum_{i=1}^k\beta_i\sum_{X \subseteq E_n}
d(H_i, X)\sigma_X.
\end{equation}
Let $K(X)$ be the weighted sum of exact homomorphism densities $d(H_i, X)$:
\begin{equation}
K(X)=n^2\sum_{i=1}^k\beta_i
d(H_i, X).
\end{equation}
The Hamiltonian notation may then be simplified,
\begin{equation}
\label{simple}
H(\sigma)=-\sum_{X \subseteq E_n}K(X)\sigma_X.
\end{equation}
The interactions $K=\{K(X): X \subseteq E_n\}$ vary with the tuning parameters $\beta_i$ and form
a Banach space $\mathcal{B}$.
Define the $\mathcal{B}$-norm of $K$ by
\begin{equation}
||K||=\sup_{e}\sum_{X: e\in X}|K(X)|.
\end{equation}
By Proposition \ref{count},
\begin{equation}
||K||\leq m(m-1)\sum_{i=1}^k
|\beta_i|.
\end{equation}
This Banach space construction
will be useful later on.
\end{proof}

Notice that $K(X)$, being a weighted sum of
$d(H_i, X)$, satisfies the finite-body property: $K(X)=0$ for
$|X|>p$ (its importance to be illustrated in the proof of Lemma
\ref{lemma}).
Moreover, summing over $G \in \mathcal{G}_n$
(\ref{psi}) is equivalent to summing over indicator functions
$\sigma$ of $E(G) \subseteq E_n$:
\begin{equation}
\label{parr} \sum_{G \in \mathcal{G}_n}\exp\left(n^2 \sum_{i=1}^k
\beta_i t(H_i, G)\right)=\sum_{\sigma} \exp
\left(-H(\sigma)\right).
\end{equation}
Define $\psi^{\beta}$ to be the limiting free energy of the random
graph model, i.e., $\psi^{\beta}=\lim_{n\rightarrow\infty}
\psi^{\beta}_n$. By (\ref{psi}), (\ref{simple}), and (\ref{parr}),
we have
\begin{equation}
\psi^{\beta}=\lim_{n\rightarrow\infty}
\frac{1}{n^2}\log \left(\sum_{\sigma}\exp\left(\sum_{X \subseteq E_n}
K(X)\sigma_X\right)\right).
\end{equation}
We normalize the sum over $\sigma$ for the ease of cluster
expansion. Henceforth $\sum^{\text{norm}}_{\sigma}$ will denote
the normalized sum and satisfy $\sum^{\text{norm}}_{\sigma} 1=1$, i.e.,
$\sum^{\text{norm}}_{\sigma}=2^{-|E_n|}\sum_{\sigma}=2^{- {n \choose 2}}\sum_{\sigma}$.
Define the partition function $W$ by
\begin{equation}
\label{partition} W=\sum^{\text{norm}}_{\sigma} \exp
\left(-H(\sigma)\right)=\sum^{\text{norm}}_{\sigma}
\exp\left(\sum_{X \subseteq E_n} K(X)\sigma_X\right).
\end{equation}
According to standard statistical mechanics,
the limiting free energy $\phi^{\beta}$ of the lattice gas model is then given by
\begin{equation}
\label{phi} \phi^{\beta}=\lim_{n\rightarrow\infty}
\frac{1}{|E_n|}\log W.
\end{equation}
We explore the relationship between the two limiting free energies
$\psi^{\beta}$ and $\phi^{\beta}$:
\begin{eqnarray}
\psi^{\beta}&=&\lim_{n\rightarrow\infty} \frac{1}{n^2} \log \left(2^{n \choose 2} W\right)=
\frac{1}{2}\left(\log 2+ \lim_{n\rightarrow\infty} \frac{1}{{n \choose 2}} \log W\right)\\
&=&\frac{1}{2}\left(\log 2+ \lim_{n\rightarrow\infty} \frac{1}{|E_n|} \log W\right)=\frac{1}{2}(\log 2+\phi^{\beta}).
\end{eqnarray}
The two interpretations of the limiting free energy are thus not
appreciably different, and we may interpret it in either way to
help with the understanding of the structure and behavior of the
limiting network.

\section{Cluster expansion}
\label{expansion} In this section we will apply cluster expansion
techniques to derive a convergent power series expansion
(high-temperature expansion) for the limiting free energy in the
case of small parameters. The cluster expansion expressions
presented here are completely rigorous for finite models, and may
be interpreted in some more sophisticated limiting sense. We begin
by introducing some combinatorial concepts. A hypergraph is a set
of sites together with a collection $\Gamma$ of nonempty subsets.
Such a nonempty set is referred to as a hyper-edge or link. Two
links are connected if they overlap. The support of a hypergraph
is the set $\cup \Gamma$ of sites that belong to some set in
$\Gamma$. A hypergraph $\Gamma$ is connected if the support of
$\Gamma$ is nonempty and cannot be partitioned into nonempty sets
with no connected links. We use $\Gamma_c$ to indicate
connectivity of the hypergraph $\Gamma_c$. Our first proposition
gives a cluster representation for the partition function of
the exponential random graph model.

\begin{proposition}
\label{par} Let $W$ be the partition function of the
exponential random graph model on $n$ vertices (\ref{partition}).
Then $W$ has a formal cluster representation
\begin{equation}
\label{tt}
W=\sum_{\Delta}\prod_{N\in \Delta}w_N,
\end{equation}
where:
\begin{itemize}
\item $\Delta$ is a set of disjoint subsets $N$ of $E_n$.

\item $w_N=\sum_{\cup\Gamma_c=N}\sum^{\text{norm}}_{\sigma|N}\prod_{X\in
\Gamma_c}(\rme^{K(X)\sigma_X}-1).$
\end{itemize}
\end{proposition}

\begin{proof}
We rewrite $\exp\left(\sum_{X \subseteq E_n} K(X)\sigma_X\right)$ as a perturbation around
zero interaction,
\begin{eqnarray}
\label{W} W=\sum^{\text{norm}}_{\sigma}\prod_{X \subseteq E_n}
\left(1+\rme^{K(X)\sigma_X}-1\right)=\sum^{\text{norm}}_{\sigma}\sum_{\Gamma}\prod_{X\in
\Gamma} \left(\rme^{K(X)\sigma_{X}}-1\right),
\end{eqnarray}
where $\Gamma$ is a set of subsets $X$ of $E_n$.

We are going to organize the sum over hypergraphs in (\ref{W}) in
the following way. Let $N$ be a possible support for a connected
hypergraph. Let $\Delta$ be a disjoint set of such sets $N$. Let
$S$ be a function that takes $N\in \Delta$ to a hypergraph with
support $N$, i.e., $\cup S(N)=N$. Then summing over hypergraphs
$\Gamma$ is equivalent to summing over $\Delta$ and functions $S$
with the appropriate property. Furthermore, the product over $N$
in $\Delta$ and the links in $S(N)$ is equivalent to the product
over the corresponding $\Gamma$. We have
\begin{equation}
W=\sum^{\text{norm}}_{\sigma}\sum_{\Delta}\sum_{S}\prod_{N\in \Delta}\prod_{X\in
S(N)}\left(\rme^{K(X)\sigma_X}-1\right).
\end{equation}
By independence, the sum over $\sigma$ can be factored over
$\Delta$.
Denote by $\sum^{\text{norm}}_{\sigma|N}$ the normalized sum restricted to subset $N$,
i.e., $\sum^{\text{norm}}_{\sigma|N}=2^{-|N|}\sum_{\sigma|N}$.
Similarly, denote by $\sum^{\text{norm}}_{\sigma|e}$ the normalized sum restriced to vertex pair $e$.
This gives
\begin{equation}
\label{comp}
W=\sum_{\Delta}\left(\prod_{e \notin \cup \Delta} \sum^{\text{norm}}_{\sigma|e}1\right)
\sum_S \prod_{N \in \Delta}\sum^{\text{norm}}_{\sigma|N}\prod_{X\in
S(N)} \left(\rme^{K(X)\sigma_{X}}-1\right).
\end{equation}
Because of the normalization, $\sum^{\text{norm}}_{\sigma|e}1=1$, (\ref{comp}) can be simplified,
\begin{equation}
W=\sum_{\Delta}\sum_S \prod_{N \in \Delta}\sum^{\text{norm}}_{\sigma|N}\prod_{X\in
S(N)} \left(\rme^{K(X)\sigma_{X}}-1\right).
\end{equation}
Rearranging the terms by the distributive law, we have
\begin{eqnarray}
W=\sum_{\Delta}\prod_{N\in
\Delta}\sum_{\cup\Gamma_c=N}\sum^{\text{norm}}_{\sigma|N}\prod_{X\in \Gamma_c}
\left(\rme^{K(X)\sigma_{X}}-1\right).
\end{eqnarray}
Our claim follows once we recall the definition of $w_N$.
\end{proof}

Notice that (\ref{tt}) has a graphical representation:
\begin{eqnarray}
W&=&\sum_{n=0}^{\infty}\frac{1}{n!}\sum_{N_1,...,N_n}\prod_{\{i,j\}}\left(1-t(N_i,N_j)\right)
w_{N_1}\cdots w_{N_n}\\
\label{exp}
&=&\sum_{n=0}^{\infty}\frac{1}{n!}\sum_{N_1,...,N_n}c\left(N_1,...,N_n\right)w_{N_1}\cdots
w_{N_n},
\end{eqnarray}
where
\begin{equation}
c\left(N_1,...,N_n\right)=\sum_{R}\prod_{\{i,j\}\in
R}\left(-t(N_i,N_j)\right),
\end{equation}
$R$ is a graph with vertex set $\{1,...,n\}$, and
\begin{eqnarray}
\label{c} t(N_i,N_j)=\left\{\begin{array}{ll}
1 & \mbox{if $N_i$ and $N_j$ overlap};\\
0 & \mbox{otherwise}.\end{array} \right.
\end{eqnarray}

This alternate expression of the partition function $W$
facilitates the application of cluster expansion ideas, which
roughly summarized, state that a sum over arbitrary graphs can be
written as the exponential of a sum over connected graphs. Taking
the logarithm of the partition function (\ref{exp}) thus
replaces the sum over graphs by the sum over connected graphs. The
log operation is physically significant in that the resulting
connected function $\log W$ is proportional to the limiting free
energy $\phi^{\beta}$ (\ref{phi}). A detailed explanation of this
phenomenon may be found, for instance, in a survey article by
Faris \cite{Faris}.

\begin{proposition}
\label{connected} Let $W$ be the partition function of the
exponential random graph model on $n$ vertices (\ref{exp}). Then
the connected function $\log W$ is given by the cluster expansion
\begin{equation}
\label{log} \log
W=\sum_{n=1}^{\infty}\frac{1}{n!}\sum_{N_1,...,N_n}C\left(N_1,...,N_n\right)w_{N_1}\cdots
w_{N_n},
\end{equation}
where
\begin{equation}
\label{C} C\left(N_1,...,N_n\right)=\sum_{R_c}\prod_{\{i,j\}\in
R_c}\left(-t(N_i,N_j)\right),
\end{equation}
and $R_c$ is a connected graph with vertex set $\{1,...,n\}$.
\end{proposition}

Now that we have derived an explicit expression for the connected
function $\log W$, we explore criteria that guarantee the
convergence of this expansion. This provides information on the
limiting free energy $\phi^{\beta}$ (\ref{phi}) and characterizes
the structure and behavior of the limiting network. The celebrated
theorem of Koteck\'{y} and Preiss says that if the interaction is
sufficiently weak, then the cluster expansion for the pinned
connected function converges.

\begin{theorem}[Koteck\'{y}-Preiss \cite{Kotecky}]
\label{KP} Consider arbitrary family of activities $v_N \geq 0$.
Suppose there are finite $B_N \geq 0$ such that for all $N_0$,
\begin{equation}
\label{ineq2} \sum_{N}t(N, N_0)v_{N} B_{N}\leq \log B_{N_0}.
\end{equation}
Then the pinned connected function has a convergent power series
expansion in the polydisc $|w_N|\leq v_N$:
\begin{equation}
\label{K} \sum_{n=1}^{\infty}\frac{1}{n!}\sum_{N_1,...,N_n:
\exists i
N_i=N}\left|C\left(N_1,...,N_n\right)\right||w_{N_1}|\cdots
|w_{N_n}|\leq v_{N}B_N.
\end{equation}
\end{theorem}

\begin{remark}
In application it is convenient to take $B_N=M^{|N|}$ with $M>1$.
With this choice of $B_N$ the Koteck\'{y}-Preiss condition is
equivalent to the condition that (\ref{ineq2}) holds for all one
point sets $N_0=\{e\}$:
\begin{equation}
\label{ineq} \sum_{N: e\in N}v_{N} M^{|N|}\leq \log M.
\end{equation}
This reduced version of the Koteck\'{y}-Preiss condition will be
used throughout the rest of this section.
\end{remark}

At first sight, the Koteck\'{y}-Preiss condition (\ref{ineq})
seems very abstract and difficult to verify. It is a weighted sum
of activities $v_N$ pinned at the vertex pair $e$, and each $v_N$
is an upper bound for $|w_N|$, whose expression involves a
hypergraph decomposition and is rather complicated by itself (cf.
Proposition \ref{par}). The following proposition gives a handy
criterion for weak interaction in the small parameter region
($\sum_{i=1}^k|\beta_i|$ small).

\begin{proposition}
\label{basic} Fix $M>1$. Take $w_N$ as in Proposition
\ref{par}. It is clear that
\begin{equation}
|w_N|\leq v_N=\sum_{\cup\Gamma_c=N}\prod_{X\in
\Gamma_c}(\rme^{|K(X)|}-1).
\end{equation}
Consider the interaction $K$ with the Banach space norm
$||K||$ (cf. Theorem \ref{con}). Suppose $\sum_{i=1}^k|\beta_i|$
is small:
\begin{equation}
\label{eps}||K|| \leq m(m-1)\sum_{i=1}^k |\beta_i| \leq \frac{\log
M (p-1)^p}{2(Mp)^p \left(1+(p-1)\log M\right)}.
\end{equation}
Then (\ref{ineq}) holds for every vertex pair $e$.
\end{proposition}

\begin{remark}
The maximal region of parameters $\{\beta_i\}$ is obtained by
setting
\begin{equation}
\log M=\frac{-p+\sqrt{5p^2-4p}}{2p(p-1)}.
\end{equation}
\end{remark}

Our ultimate goal is to examine convergence of the limiting free
energy $\phi^{\beta}$ (\ref{phi}), which is proportional to the
connected function $\log W$ (\ref{log}). As the
Koteck\'{y}-Preiss result concerns convergence of the pinned
connected function (\ref{K}), it appears to be inapplicable. However, our
next theorem shows that pinning is actually the central ingredient
that ties these two seemingly unrelated issues together.

\begin{theorem} [Main Theorem]
\label{main} Fix $M>1$. Consider the interaction $K$ with the
Banach space norm $||K||$ (cf. Theorem \ref{con}). Suppose
$\sum_{i=1}^k |\beta_i|$ is small (\ref{eps}). Then the limiting
free energy $\phi^{\beta}$ (\ref{phi}) is analytic in $\beta$ and
the rate of convergence is uniform.
\end{theorem}

\begin{proof}
By Proposition \ref{basic}, (\ref{ineq}) holds for every vertex
pair $e$. The pinned connected function (\ref{K}) thus converges
absolutely by Theorem \ref{KP}, which further implies
\begin{eqnarray}
|\log W|&\leq& \sum_{N\subseteq E_n}
\sum_{n=1}^{\infty}\frac{1}{n!}\sum_{N_1,...,N_n: \exists i
N_i=N}\left|C\left(N_1,...,N_n\right)\right||w_{N_1}|\cdots
|w_{N_n}|\\
&\leq&\sum_{N\subseteq E_n} v_{N}M^{|N|}\leq\sum_{e\in
E_n}\sum_{N: e\in N} v_{N}M^{|N|} \leq |E_n|\log M.
\end{eqnarray}
We conclude that the
limiting free energy $\phi^{\beta}$ is absolutely convergent and bounded above by $\log M$.
\end{proof}

The rest of this section is devoted to the proof of Proposition
\ref{basic}. The weighted activity sum in the Koteck\'{y}-Preiss
weak interaction condition (\ref{ineq}) is rewritten as a power
series, whose terms are then shown to be exponentially small under
(\ref{eps}) by a series of lemmas.

\vskip.1truein

\noindent \textit{Proof of Proposition \ref{basic}.} We notice
that when $||K||$ is small (say $||K||\leq \frac{1}{2}$),
$\rme^{|K(X)|}-1\leq 2|K(X)|$ by the mean value theorem. For
$\cup\Gamma_c=N$, $|N|\leq \sum|X|$ with $X$ in $\Gamma_c$. We
have
\begin{eqnarray}
\sum_{N: e\in N}v_N M^{|N|}&\leq&\sum_{N: e\in N}\sum_{\cup\Gamma_c=N}M^{|N|}\prod_{X\in \Gamma_c}2|K(X)|\\
&\leq&\sum_{\Gamma_c: e\in \cup\Gamma_c}\prod_{X\in
\Gamma_c}2|K(X)|M^{|X|}.
\end{eqnarray}
We say that a hypergraph $\Gamma_c$ is rooted at the vertex pair
$e$ if $e\in \cup\Gamma_c$. Let $a_n(e)$ be the contribution of
all connected hypergraphs with $n$ links that are rooted at $e$,
\begin{equation}
\label{ann} a_n(e)=\sum_{e\in \cup\Gamma_c:
|\Gamma_c|=n}\prod_{X\in \Gamma_c}2|K(X)|M^{|X|}.
\end{equation}
Then
\begin{equation}
\sum_{N: e\in N}v_N M^{|N|}\leq \sum_{n=1}^{\infty}\sup_{e\in
E_n}a_n(e):=\sum_{n=1}^{\infty} a_n.
\end{equation}
It seems that once we show that $a_n$ is exponentially small, the
power series above will converge, and our claim might follow. To
estimate $a_n$, we relate to some standard combinatorial facts
\cite{Minlos}. \qed

\begin{lemma}
\label{lemma} Let $a_n$ be the supremum over $e$ of the
contribution of connected hypergraphs with $n$ links that are
rooted at $e$. Then $a_n$ satisfies the recursive bound
\begin{equation}
a_n\leq 2||K||M^p\sum_{k=0}^{p}{p \choose k}
\sum_{a_{n_1},...,a_{n_k}: n_1+...+n_k+1=n}a_{n_1}\cdots a_{n_k}
\end{equation}
for $n\geq 1$, where ${p \choose k}$ is the binomial coefficient.
\end{lemma}

\begin{proof}
We linearly order the vertex pairs $e$ in $E_n$ and also linearly
order the subsets $X$ of $E_n$. For a fixed but arbitrarily
chosen $e$ in $E_n$, we examine (\ref{ann}). Write
$\Gamma_c=\{X_1\}\cup \Gamma^1_c$, where $X_1$ is the least $X$ in
$\Gamma_c$ with $e \in X_1$. As $K(X_1) \neq 0$ only for
$|X_1|\leq p$,
\begin{equation}
a_n(e)\leq 2||K||M^p \sum_{\Gamma^1_c}\prod_{X\in
\Gamma^1_c}2|K(X)|M^{|X|}.
\end{equation}
The remaining hypergraph $\Gamma^1_c$ has $n-1$ subsets and breaks
into $k: k\leq p$ connected components (which again follows from
the finite-body property of $K$). Say they are
$\Gamma_1,...,\Gamma_k$ of sizes $n_1,...,n_k$, with
$n_1+...+n_k=n-1$. For each component
$\Gamma_i$, there is a least vertex pair $e_i$ through which it is
connected to $X_1$, i.e., $e_i \in X_1$ is a root of $\Gamma_i$. As
$|X_1|\leq p$ and $\Gamma^1_c$ consists of $k$ components,
the number of possible choices for the root locations is at most ${p \choose k}$. We have
\begin{equation}
a_n(e)\leq 2||K||M^p \sum_{k=0}^{p} {p \choose k}
\sum_{a_{n_1},...,a_{n_k}: n_1+...+n_k+1=n}a_{n_1}\cdots a_{n_k}.
\end{equation}
Our inductive claim follows by taking the supremum over all $e$ in
$E_n$. Finally, we look at the base step: $n=1$. In this simple
case, as reasoned above, we have
\begin{equation}
\label{base}
a_1=\sup_{e\in E_n}\sum_{e\in \cup\Gamma_c:
|\Gamma_c|=1}\prod_{X\in \Gamma_c}2|K(X)|M^{|X|}\leq2||K||M^{p},
\end{equation}
and this verifies our claim.
\end{proof}

Clearly, $\sum_{N: e\in N}v_N M^{|N|}$ will be bounded above by
$\sum_{n=1}^{\infty}\bar{a}_n$, if
\begin{equation}
\label{a} \bar{a}_n= 2||K||M^p\sum_{k=0}^p {p \choose
k}\sum_{\bar{a}_{n_1},...,\bar{a}_{n_k}:
n_1+...+n_k+1=n}\bar{a}_{n_1}\cdots \bar{a}_{n_k}
\end{equation}
for $n\geq 1$, i.e., equality is obtained in the above lemma.
Observe that when $n=1$, the empty product (corresponding to $k=0$) is the only nonzero term on the right side of (\ref{a}),
so
we have 
$\bar{a}_1=2||K||M^p$, matching the bound in (\ref{base}).

\begin{lemma}
\label{lemma2} Consider the coefficients $\bar{a}_n$ that bound
the contributions of connected and rooted hypergraphs with $n$
links. Let $w=\sum_{n=1}^{\infty}\bar{a}_n z^n$ be the generating
function of these coefficients. The recursion relation (\ref{a})
for the coefficients is equivalent to the formal power series
generating function identity
\begin{equation}
\label{id} w=2||K||M^pz(1+w)^p.
\end{equation}
\end{lemma}

\begin{proof}
Notice that $(1+w)^p=\sum_{k=0}^p {p \choose k} w^k$, thus
\begin{eqnarray}
w=2||K||M^pz\sum_{k=0}^p {p \choose k} w^k.
\end{eqnarray}
Writing completely in terms of $z$, we have
\begin{equation}
\label{equ}
\sum_{n=1}^{\infty}\bar{a}_n z^n=2||K||M^p\sum_{k=0}^p {p \choose
k} \sum_{\bar{a}_{n_1},...,\bar{a}_{n_k}:
n_1+...+n_k+1=n}\bar{a}_{n_1}\cdots \bar{a}_{n_k}z^n.
\end{equation}
We compare the coefficient of $z$ on both sides of (\ref{equ}): On the left it is given by $\bar{a}_1$, and 
on the right it is given by $2||K||M^p$ times the empty product, thus
$\bar{a}_1=2||K||M^p$.
Our general claim follows from term-by-term comparison.
\end{proof}

\begin{lemma}
\label{lemma3} If $w$ is given as a function of $z$ as a formal
power series by the generating function identity (\ref{id}), then
this power series has a nonzero radius of convergence $|z|\leq
\frac{(p-1)^{p-1}}{2||K||(Mp)^p}$.
\end{lemma}

\begin{proof}
Without loss of generality, assume $z\geq 0$. Set
$z_1=2||K||M^pz$. Solving (\ref{id}) for $z_1$ gives
$z_1=w/(1+w)^p$. As $z_1$ goes from $0$ to $(p-1)^{p-1}/{p^p}$,
the $w$ values range from $0$ to $1/(p-1)$.
\end{proof}

\noindent \textit{Proof of Proposition \ref{basic} continued.} We
notice that in the above lemma,
$w=\sum_{n=1}^{\infty}\bar{a}_nz^n=1/(p-1)$ corresponds to
$2||K||M^pz=(p-1)^{p-1}/{p^p}$, which implies that for each $n$,
\begin{equation}
\bar{a}_n \leq
\left(2||K||(Mp)^p\right)^n\left(p-1\right)^{-\left(1+(p-1)n\right)}.
\end{equation}
Gathering all the information we have obtained so far,
\begin{eqnarray}
\sum_{N: e\in N}v_N M^{|N|}&\leq&
\sum_{n=1}^{\infty}\left(2||K||(Mp)^p\right)^n\left(p-1\right)^{-\left(1+(p-1)n\right)}
\\&=&\frac{\frac{2||K||(Mp)^p}{(p-1)^p}}{1-\frac{2||K||(Mp)^p}{(p-1)^{p-1}}}\leq
\log M.
\end{eqnarray}
by (\ref{eps}). \qed

\section{Concluding remarks}
\label{conclusion} This paper reveals a deep connection between
random graphs and lattice gas (Ising) systems, making the
exponential random graph model treatable by cluster expansion
techniques from statistical mechanics. We show that any
exponential random graph model may alternatively be viewed as a
lattice gas model with a finite Banach space norm and derive a
convergent power series expansion (high-temperature expansion) for
the limiting free energy in the case of small parameters. Since
the free energy is the generating function for the expectations of
other random variables, this characterizes the structure and
behavior of the limiting network in this parameter region. We hope
this rigorous expansion will provide insight into the limiting
structure of exponential random graphs in other parameter regions
and shed light on the application of renormalization group ideas
to these models.

\ack The author is grateful to Persi Diaconis, Sourav Chatterjee,
and Charles Radin for introducing her to the exciting subject of
random graphs and for their many enlightening and encouraging
comments. She also thanks her PhD advisor Bill Faris for his
continued help and support. The author appreciated the opportunity
to talk about this work in the 2011 workshop in Dynamical
Gibbs-non-Gibbs Transitions at EURANDOM, organized by Aernout van
Enter, Roberto Fernandez, Frank den Hollander, and Frank Redig.

\end{document}